\documentclass[conference]{IEEEtran}

\usepackage{theorem}

\usepackage{amsmath}
\usepackage{amsfonts}
\usepackage{latexsym}
\usepackage{amssymb}
\usepackage{algorithm2e}
\usepackage{graphics}
\usepackage{graphicx}
\usepackage{subfigure}
\usepackage{enumerate}

\usepackage{mathrsfs} % Define \mathscr, a script font

\usepackage{cite} 
\usepackage{upref}

\usepackage{theorem}
\usepackage{psfrag}
\usepackage{color}
\usepackage{tikz,pgfplots}

\usepackage{geometry}
\IEEEoverridecommandlockouts
\title{\Huge
{Finite-Length Analysis of\\Polar Secrecy Codes for Wiretap Channels}}
\author{\large%
Hessam Mahdavifar and Fariba Abbasi
\thanks{H.\ Mahdavifar is with the Department of Electrical and Computer Engineering, Northeastern University, Boston, MA 02115  (email: h.mahdavifar@northeastern.edu.) F.\ Abbasi is with the Department of Electrical Engineering and Computer Science, University of Michigan, Ann Arbor, MI 48104 (email: fabbasia@umich.edu.)}
\thanks{This work was supported in part by the Center for Ubiquitous Connectivity (CUbiC), sponsored
by Semiconductor Research Corporation (SRC) and Defense Advanced Research Projects
Agency (DARPA) under the JUMP 2.0 program.}
}

\newcommand{\cA}{{\cal A}} 
\newcommand{\cB}{{\cal B}}
\newcommand{\cC}{{\cal C}} 

\newcommand{\cE}{{\cal E}} 

\newcommand{\cG}{{\cal G}}

\newcommand{\cR}{{\cal R}}

\newcommand{\cU}{{\cal U}}

\DeclareMathAlphabet{\mathbfsl}{OT1}{ppl}{b}{it} %{OT1}{cmr}{bx}{it}

\newcommand{\bU}{\mathbfsl{U}} 
\newcommand{\bV}{\mathbfsl{V}}

\newcommand{\bZ}{\mathbfsl{Z}}

\newcommand{\bu}{\mathbfsl{u}} 
\newcommand{\bv}{\mathbfsl{v}}
 
\newcommand{\bx}{\mathbfsl{x}}

\newcommand{\eee}{\mathbfsl{e}} 

\newcommand{\muo}{\overline{\mu}} 
\newcommand{\muu}{\underline{\mu}} 

\def\QEDclosed{\mbox{\rule[0pt]{1.3ex}{1.3ex}}} % for a filled box
\def\QED{\QEDclosed} % default to closed

\newcommand{\be}[1]{\begin{equation}\label{#1}}
\newcommand{\ee}{\end{equation}} 
\newcommand{\eq}[1]{(\ref{#1})}

\renewcommand{\leq}{\leqslant}
 
\renewcommand{\geq}{\geqslant}

\renewcommand{\Bbb}{\mathbb}
 
\newcommand{\N}{{\Bbb N}}
\newcommand{\R}{{\Bbb R}} 

\newcommand{\F}{{\Bbb F}}

\newcommand{\Tref}[1]{Theo\-rem\,\ref{#1}}

\newcommand{\Cref}[1]{Co\-ro\-lla\-ry\,\ref{#1}}

\newcommand{\deff}{\mbox{$\stackrel{\rm def}{=}$}}

\newcommand{\Gm}{G^{\otimes m}}

\newcommand{\shalf}{\mbox{\raisebox{.8mm}{\footnotesize $\scriptstyle 1$}
\footnotesize$\!\!\! / \!\!\!$ \raisebox{-.8mm}{\footnotesize
$\scriptstyle 2$}}}

\theoremstyle{plain} 
\theorembodyfont{\normalfont\slshape}

\newtheorem{thm}{Theorem\hspace{-1pt}} 
\newenvironment{theorem}
{\begin{thm}\hspace*{-1ex}{\bf.}}{\end{thm}}

\newtheorem{lem}[thm]{Lemma\hspace{-.75pt}}
\newenvironment{lemma}{\begin{lem}\hspace*{-1ex}{\bf.}}{\end{lem}}

\newtheorem{prop}[thm]{Proposition$\!$}
\newenvironment{proposition}{\begin{prop}\hspace*{-1ex}{\bf.}}{\end{prop}}

\newtheorem{cor}[thm]{Corollary$\!$}
\newenvironment{corollary}{\begin{cor}\hspace*{-1ex}{\bf.}}{\end{cor}}

\newtheorem{defn}{Definition$\!$}
\newenvironment{definition}{\begin{defn}\hspace*{-1ex}{\bf.}}{\end{defn}}

\begin{document}

\vspace{10mm}
\maketitle

\begin{abstract}
In a classical wiretap channel setting, Alice communicates with Bob through a main communication channel, while her transmission also reaches an eavesdropper Eve through a wiretap channel. In this paper, we consider a general class of polar secrecy codes for wiretap channels and study their finite-length performance. In particular, bounds on the normalized mutual information security (MIS) leakage, a fundamental measure of secrecy in information-theoretic security frameworks, are presented for polar secrecy codes. The bounds are utilized to characterize the finite-length scaling behavior of polar secrecy codes, where scaling here refers to the non-asymptotic behavior of both the gap to the secrecy capacity as well as the MIS leakage. Furthermore, the bounds are shown to facilitate characterizing numerical bounds on the secrecy guarantees of polar secrecy codes in finite block lengths of practical relevance, where directly calculating the MIS leakage is in general infeasible. 

\end{abstract}

%=======================================================================%
%                                                                       %
%     2. Preliminaries                    %              
%                                                                       %
%=======================================================================%
\section{Introduction}
\label{sec:intro}

The wiretap channel model was first introduced by Wyner in 1975 \cite{Wyner}. Wyner's seminal work paved the way for an entire area of research encompassing hundreds of papers by now, often referred to as physical layer security. In Wyner's setting, demonstrated in Figure\,\ref{fig:WynerSetting}, a transmitter (Alice) communicates with a legitimate receiver (Bob) through a main communication channel $W^*$, while her transmission also reaches an eavesdropper Eve through a wiretap channel $W$. The goal is to design a coding scheme that makes it possible for Alice to communicate both reliably and securely. Reliability is measured in terms of Bob's probability of error in estimating the message $\bU$, denoted by $\hat{\bU}$, same as in a classical channel coding problem. The security condition is expressed in terms of the mutual information between the message $\bU$ and Eve's observation $\bZ$, often referred to as the mutual information security (MIS) leakage. More specifically, the \textit{weak} secrecy condition, introduced by Wyner \cite{Wyner}, requires that $\lim_{k \rightarrow \infty} I(\bU;\bZ)/k = 0$, where $k$ is the length of the message $\bU$. Wyner characterized the fundamental limit on the communication rate in this regime, referred to as the \textit{secrecy capacity}. Later, Maurer extended the security condition to \text{strong} secrecy, that is, to require that $\lim_{k \rightarrow \infty} I(\bU;\bZ) = 0$ \cite{Maurer}. Remarkably, Maurer showed that the secrecy capacity remains the same under strong secrecy. 

%Note that, with slight abuse of terminology, the entire setup as well as the channel between Alice and Eve are referred to as the wiretap channel, which should be clear from the context. 

\begin{figure}[t]
\centering
\includegraphics[width=0.9\linewidth]{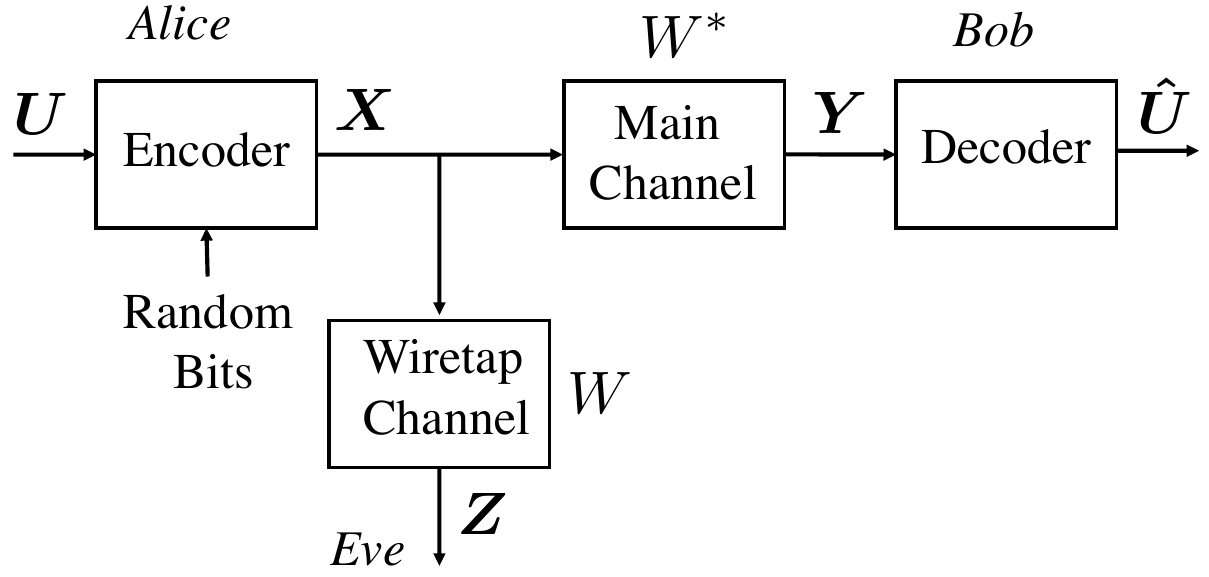}
%{Performance163.png}
\caption{Wyner's wiretap channel setting. }

\label{fig:WynerSetting}
\end{figure}

Polar secrecy codes, built upon Ar{\i}kan's polar codes \cite{Arikan}, were first introduced in \cite{MahdavifarSecrecy} to achieve the secrecy capacity of wiretap channels when $W$ and $W^*$ are both binary-input memoryless symmetric (BMS) channels and $W$ is \textit{degraded with respect to} $W^*$. The latter means that $W$ can be written as a cascade of $W^*$ with another channel. Under this assumption the secrecy capacity is $C_s = C(W^*) - C(W)$, where $C(.)$ is the channel capacity \cite{Wyner}. Note that the results in \cite{MahdavifarSecrecy} are established under the weak secrecy condition. A scheme together with empirical observations are provided in \cite{MahdavifarSecrecy} for strong secrecy, however, it remains an open problem whether \textit{plain} polar codes are strong enough to achieve the secrecy capacity with strong secrecy. 

There have been various extensions of the scheme in \cite{MahdavifarSecrecy}, e.g., to establish different variations of secrecy \cite{bloch2013strong,chou2015polar,gunlu2020randomized}. It has also found applications in quantum communication \cite{wilde2013polar}, and quantum key distribution \cite{jouguet2012high,fang2022improved}. Also, an extension of this scheme together with a \textit{chaining} constructions is used in \cite{Sasoglu} to establish secrecy capacity-achieving property, under the same assumptions as in \cite{MahdavifarSecrecy}, with strong secrecy. Due to the nature of the chaining construction in \cite{Sasoglu}, the results are only asymptotic and the techniques are not \textit{strong} enough to characterize the \textit{scaling} behavior of the chained polar secrecy codes, i.e., how fast the gap-to-secrecy capacity and the MIS leakage of the scheme approach zero. 

In this paper, we go back to the basic scheme of \cite{MahdavifarSecrecy}, which only involves \textit{plain} polar coding, i.e., no concatenation or chaining constructions are applied on top of the scheme. We provide new bounds on the MIS leakage of polar secrecy codes. We show how these bounds, together with leveraging existing results on finite-length scaling of polar codes by Hassani \textit{et al.} \cite{hassani}, lead to a characterization of finite-length scaling of polar secrecy codes. In other words, we characterize how the gap to secrecy capacity of polar secrecy codes as well as their MIS leakage scale with the block length as the rate approaches the secrecy capacity. The bounds presented in this paper not only enables us to characterize the finite-length scaling of polar secrecy codes, but also to numerically bound the actual MIS leakage in practical settings. 

The rest of this paper is organized as follows. In Section\,\ref{sec:two} some preliminaries are presented for polar codes and secrecy codes. In Section\,\ref{sec:three}, we present new bounds on the MIS leakage of polar secrecy codes. Finite-length scaling behavior of polar secrecy codes are characterized in Section\,\ref{sec:four}, and some numerical results are presented in Section\,\ref{sec:five}.

%=======================================================================%
%                                                                       %
%     2. Preliminaries                    %              
%                                                                       %
%=======================================================================%
\section{Preliminaries}
\label{sec:two}

\subsection{Notation Convention}

First, let us briefly introduce the notations. For $n \in \N$, let $[n]:=\{1,2,\dots,n\}$. The parameter $n$ is reserved for the code block length throughout the paper. Vectors are denoted by bold lowercase letters. For a vector $\bv = (v_1,v_2,\dots,v_n)$, and $\cA \in [n]$, let $\bv_{\cA}$ denote the sub-vector of $\bv$ with entries $v_i$'s for $i \in \cA$. Following the convention, uppercase letters represent random variables and lowercase letters represent the instances of random variables, e.g., $\bU = (U_1,U_2,\dots,U_n)$ represents a vector of random variables $U_i$'s and $\bu$ represents the instances $u_i$'s of random variables $U_i$'s. 

Let $G$ denote the $2 \times 2$ matrix $\left[ 
\begin{array}{@{\hspace{0.50ex}}c@{\hspace{1.25ex}}c@{\hspace{0.50ex}}}
1 & 0\\
1 & 1\\ 
\end{array}\right]$ throughout this paper. Let $n = 2^m$, for some positive integer $m$. Then $G^{\otimes m}$ denotes the $\text{$m$-th Kronecker power of $G$}$, which is an $n \times n$ matrix. 

%Throughout the paper, $W$ represents the wiretap channel and $W^{*}$ represent the main channel, in the wiretap channel setting demonstrated in Figure \ref{}. Note that with slight abuse of terminology, wiretap channel may refer to both the channel between Alice and Eve, as well as the entire setting shown in Figure \ref{}, which should be clear from the context. 

\subsection{Polar Codes}

Polar codes belong to a general family of codes whose generator matrices, for the block length $n=2^m$, are sub-matrices of $\Gm$. Any code $\cC$ belonging to this family can be essentially specified by a set of \textit{information} bits $\cA \subseteq [n]$. In other words, $\cA$ denotes the set of indices of the rows of $\Gm$ to be included in the generator matrix of $\cC$. Alternatively, $\cC$ can be specified by $\cB = \cA^{c}$, where, in the context of polar coding, $\cB$ is referred to as the set of frozen indices. To unify the terminology, we simply refer to any code in this family as a polar code. More specifically, we refer to $\cC$ as an $(n,k)$ polar code associated with $\cA$, where $k = |\cA|$. %Note that with this terminology, RM codes also belong to this family of polar codes. 

The encoder for an $(n,k)$ binary polar code $\cC$ associated with $\cA$ is as follows. The input to the encoder is $\bu \in \{0,1\}^k$. The encoder first forms the vector $\bv \in \{0, 1\}^n$, by setting $\bv_\mathcal{A} = \bu$ and $\bv_\mathcal{B} = \mathbf{0}$. Then, the encoder outputs $\bx = \bv \Gm$. 

Let us recall the terminology of \textit{good} bit-channels from the polar coding literature. In general, the polarization transform, introduced by Ar{\i}kan \cite{Arikan}, with $n$ inputs and $n$ outputs over a BMS channel $W$ is split into $n$ synthesized bit-channels, denoted by $W_1,W_2,\dots,W_n$. Roghly speaking, the $i$-th bit-channel $W_i$ is the channel that the $i$-th input bit observes in the polarization transform. The \textit{quality} of these bit-channels, i.e., a metric to measure how \textit{good} they are, is defined through certain underlying parameters of the bit-channels, e.g., the Bhattacharyya parameter. More specifically, the set of \textit{good} bit-channels in the polarization transform, with an underlying threshold $\epsilon$, is denoted by $\cG_n(W,\epsilon)$ and is defined as
\be{good-def}
\cG_n(W,\epsilon) \,\deff\, \{i: Z(W_i) < \epsilon\},
\ee
where $Z(.)$ is the Bhattacharyya parameter of the channel. 

\subsection{Polar Secrecy Codes}

The general form of the scheme in \cite{MahdavifarSecrecy}, which we refer to as a polar secrecy code (or the plain scheme, as referred to earlier), is specified by splitting $[n]$, where $n$ is the block length, into three subsets $\cA$, $\cB$, and $\cR$ of indices for information bits, frozen bits, and random bits, respectively. Let $|\cA| = k$, and $|\cR| = r$. Consequently, $|\cB| = n - k - r$. In particular, a polar secrecy code is defined through its encoder, formally defined as follows:
\begin{definition}
\label{secrecy_code_def}
The encoder for the polar secrecy code of length $n$ associated with $(\cA,\cR)$ is a function $\cE: \{0, 1\}^k \times \{0, 1\}^r \rightarrow \{0, 1\}^n$. It accepts as input a message $\bu \in \{0, 1\}^k$ and a vector $\eee \in \{0, 1\}^r$. It is assumed that the entries of $\eee$ are selected independently and uniformly at random by the encoder. The encoder first forms the vector $\bv \in \{0, 1\}^n$, by setting
$\bv_{\cA} = \bu$, $\bv_{\cR} = \eee$, and $\bv_{\cB} = \mathbf{0}$. The encoder then outputs $\cE(\bu,\eee) := \bv \Gm$. The secrecy rate of this scheme is $R_s := \frac{k}{n}$. 
\end{definition}

%A certain method for setting the subsets $\cA$, $\cB$, and $\cR$ in the polar secrecy code is proposed and studied \cite{MahdavifarSecrecy} leveraging the notion of good bit-channels in the polarization transform. The idea behind this method is that the random bits are transmitted over the bit-channels that are \textit{good} for both Eve and Bob, information bits over those bit-channels that are \textit{good} for Bob but not \textit{good} for Eve, and zeros over those bit-channels that are not good for neither Bob nor Eve. The intuition behind why this works is that random bits will flood the \textit{capacity} of Eve's channel, thereby only a small fraction of information bits $\bu$ may be leaked to Eve through the wiretap channel. While this approach is proven to be effective in achieving the secrecy capacity under the weak secrecy condition, it falls short of achieving the secrecy capacity with strong secrecy \cite{MahdavifarSecrecy}. 

An upper bound on the MIS leakage of polar secrecy codes is presented in \cite[Lemma 15]{MahdavifarSecrecy} which we recall here under a slight variation:

\begin{proposition}
\label{prop-UB}
The following upper bound holds on the MIS leakage $I(\bU; \bZ)$ of a polar secrecy code associated with $(\cA,\cR)$, as defined in Definition\,1:
\begin{equation}
I(\bU; \bZ) \leq \min (\sum_{i \in {\mathcal{R}}^c} C_{i}, k), 
\end{equation}
where $C_i$ is the capacity of the $i$-th bit-channel $W_i$ in the polarization transform of $W$. 
\end{proposition}

%=======================================================================%
%                                                                       %
%     2. Preliminaries                    %              
%                                                                       %
%=======================================================================%
\section{Bounds on the MIS Leakage\\ of Polar Secrecy Codes}
\label{sec:three}

In this section, we first present a new lower bound on the MIS leakage $I(\bU; \bZ)$ for polar secrecy codes in wiretap settings. 

Note that while the upper bound in Proposition \ref{prop-UB} holds regardless of the distribution on $\bU$, our new lower bound holds when the input message $\bU$ has a uniform distribution. Note that without any assumption on the distribution of $\bU$, $I(\bU; \bZ)$ can be as small as zero. However, in practice, the message bits often come from a uniform distribution which is a common assumption in designing digital communication systems. 

Before proceeding to demonstrate and prove our lower bound on the MIS leakage, we present the following lemma which will be used in the proof of the lower bound, and could be of independent interest. The proof of lemma is omitted due to space constraints. 

\begin{lemma}
\label{lem-LB}
Consider the polarization transform of length $n$ over the BMS channel $W$ with input $\bV$, and output $\bZ$. Then for a subset of indices $\mathcal{D} \subseteq \{1,2,\dots,n\}$, we have
\begin{equation}
\label{lem-LB1}
I(\bV_\mathcal{D};\bZ|\bV_{\mathcal{D}^c}) \geq \sum_{i \in \mathcal{D}} C_i,
\end{equation}
where $C_i$ is the capacity of the $i$-th bit-channel $W_i$. 
\end{lemma}

\begin{proposition}
\label{prop-LB}
The following lower bound holds on the MIS leakage $I(\bU; \bZ)$ of a polar secrecy code associated with $(\cA,\cR)$, as defined in Definition\,1:
\begin{equation}
I(\bU; \bZ) \geq \max(\sum_{i \in {\mathcal{A}}} C_{i}  + \sum_{j \in \mathcal{R}} C_j -r , 0),\label{LowerboundMIS}
\end{equation}
where $r = |\mathcal{R}|$, and $C_i$ is the capacity of the $i$-th bit-channel $W_i$.
\end{proposition}

\begin{proof}
First, we have
\begin{equation}
\begin{aligned}
I(\bU; \bZ) & =I(\bV_\mathcal{A}; \bZ | \bV_\mathcal{B} = \mathbf{0}) \\
& =I(\bV_\mathcal{A}; \bZ | \bV_\mathcal{B} = \bv_\mathcal{B}),
\end{aligned}
\end{equation}
for any $\bv_\mathcal{B} \in \{0,1\}^{n-k-r}$, where the equality holds by the symmetry of the channel. Hence, one can assume that $\bv_\mathcal{B}$ is an instance of $\bV_\mathcal{B}$ of i.i.d. uniform binary random variables with the observation given to Eve. In other words, we have
\begin{equation}
I(\bU; \bZ) = I(\bV_\mathcal{A}; \bZ | \bV_\mathcal{B}). 
\label{Lemma10}
\end{equation}
The rest of the proof is by a series of equalities and inequalities as follows. 
\begin{equation}
\begin{aligned}
I(\bU; \bZ) &  \overset{(a)}{=}  
I(\bV_\mathcal{A}; \bZ|\bV_\mathcal{B}) \\
& \overset{(b)}{=}   I(\bV_{\mathcal{A} \cup \mathcal{R} }; \bZ|\mathbf{V_B}) - I(\bV_\mathcal{R}; \bZ|\mathbf{V_{\mathcal{A} \cup \mathcal{B} }}) \\
& \overset{(c)}{\geq} \sum_{i \in {\mathcal{A}}} C_{i}  + \sum_{j \in \mathcal{R}} C_j - I(\bV_\mathcal{R}; \bZ|\mathbf{V_{\mathcal{A} \cup \mathcal{B} }}) \\
& \overset{(d)}{\geq} \sum_{i \in {\mathcal{A}}} C_{i}  + \sum_{j \in \mathcal{R}} C_j -r,
\end{aligned}
\end{equation}
where $(a)$ holds as in \eqref{Lemma10}, $(b)$ holds by the chain rule of mutual information, $(c)$ holds by Lemma 2, and $(d)$ holds since $\bV_\mathcal{R}$ is a binary vector of length $r$. Furthermore, we know that the mutual information $I(\bU; \bZ)$ is a non-negative value. This concludes the proof of the proposition. 
\end{proof}

In the next corollary, we rearrange the terms in the lower and upper bounds in order to unify them and present them in terms of the capacity of the wiretap channel $C$. Note that the sum of capacities of all bit-channels in the polarization transform of length $n$ equals $nC$, i.e., 
\begin{equation}
\label{sum-capacity}
\sum_{i \in {\mathcal{A}}} C_{i} + \sum_{j \in \mathcal{B}} C_j + \sum_{l \in \mathcal{R}} C_l = \sum_{i=1}^n  C_i = nC. 
\end{equation}
This leads to the following corollary which combines the results from Proposition \ref{prop-UB} and Proposition \ref{prop-LB}. 
\begin{corollary}
For the MIS leakage $I(\bU; \bZ)$ we have
\begin{equation}
\label{eq-cor}
n(C - \frac{r}{n}) -  \sum_{i \in \mathcal{B}} C_i \leq I(\bU; \bZ) \leq n(C- \frac{r}{n}) + \sum_{l \in \mathcal{R}} (1-C_l).
\end{equation}
\label{cor-bounds}
\end{corollary}
Note that the term $C - \frac{r}{n}$ is the gap to capacity of a polar code of length $n$ associated with  $\mathcal{R}$. There is already an extensive literature on studying this quantity, also referred to as finite-length scaling of polar codes, started by Hassani \textit{et al.} in \cite{hassani}. This is the main reason for representing the bounds as in Corollary \ref{cor-bounds} using which together with the already known results on the finite-length scaling of polar codes we will characterize the scaling behavior of the polar secrecy codes.

%=======================================================================%
%                                                                       %
%     5. Preliminaries                    %              
%                                                                       %
%=======================================================================%

\section{Finite-Length Scaling of Polar Secrecy Codes}
\label{sec:four}

First, let us recall the known results on finite-length scaling of polar codes. In general, a widely accepted conjecture is that for a sequence of polar codes guaranteeing a probability of error upper bounded by a fixed value $P_e$ and of increasing length $n$, the gap to capacity $C - \frac{r}{n}$, where $\frac{r}{n}$ is the rate and $C$ is the capacity, scales as $n^{-1/\mu}$, for some $\mu > 2$ referred to as the scaling exponent, i.e., 
\begin{equation}
\label{polar-scaling}
C - \frac{r}{n} = \alpha n^{-1/\mu} + o(n^{-1/\mu}),
\end{equation}
for some constant $\alpha \in \R^+$, where $\mu$ depends on the underlying channel. This problem was first studied by Hassani \textit{et al.} \cite{hassani}. It is shown in \cite{hassani} that there exists $\overline{\mu}$ and $\underline{\mu}$ such that
\begin{equation}
\label{polar-scaling2}
\underline{\alpha} n^{-1/\muu} \leq C - \frac{r}{n} \leq \overline{\alpha} n^{-1/\muo},
\end{equation}
where $\underline{\alpha}, \overline{\alpha} \in \R^+$, for large enough, yet finite, values of $n$. Furthermore, numerical values for $\muu$ and $\muo$ are characterized in \cite{hassani}. %Also, $\mu = 3.627$ is computed for the binary erasure channel (BEC) in \cite{hassani}, and it is conjectured that this is the best scaling exponent for BMS channels, i.e., BEC polarizes the \textit{fastest} among BMS channels. 

\subsection{Bounds on the Scaling Exponent}
\label{sec:four-A}

In this section, we consider polar secrecy codes associated with $(\cA,\cR)$ set as follows:
\be{AR-def}
\cR = \cG_n(W,\epsilon),\ \text{and}\ \cA = \cG_n(W^*,\epsilon) \setminus \cR,
\ee
where $\epsilon = P_e/n$ and $P_e$ is a fixed value representing a bound on the probability of error at Bob's side. Note that by the degraded assumption of the wiretap setting, we have $\cG_n(W,\epsilon) \subseteq \cG_n(W^*,\epsilon)$ \cite[Lemma 4.7]{Korada}. Recall that the number of information bits in this scheme is $k = |\cA|$, and the secrecy rate is $R_s = k/n$. 

Let $\muu$ and $\muo$ denote the lower bound and upper bound on the scaling exponent of polar codes over the wiretap channel $W$. Also, let $\underline{\alpha}$ and $\overline{\alpha}$ denote the corresponding constants. Also, let $\muu^*$ and $\muo^*$, $\underline{\alpha}^*$ and $\overline{\alpha}^*$ be defined similarly for the main channel $W^*$. 

Now, utilizing \eq{polar-scaling2} one can express the lower and upper bound on the MIS leakage of the above polar secrecy code, as presented in Corollary \ref{cor-bounds}, as follows: 
\begin{equation}
\label{eq-cor2}
I(\bU; \bZ) \geq \underline{\alpha} n^{1-1/\muu} + o(n^{1-1/\muu}) -  \sum_{i \in \mathcal{B}} C_i,
\end{equation}
and
\begin{equation}
\label{eq-cor3}
I(\bU; \bZ) \leq \overline{\alpha} n^{1-1/\muo} + o(n^{1-1/\muo}) + \sum_{l \in \mathcal{R}} (1-C_l).
\end{equation}
Note that we prefer to use the small-$o$ notation to avoid repeatedly saying the inequalities hold for large enough $n$. 

%We show in the next theorem that the term $\sum_{i \in \mathcal{A}} C_i \leq I(\bU, \bZ)$ in the lower bound as well as the term $\sum_{l \in \mathcal{R}} (1-C_l)$ are both negligible compared to the main term coming from the finite-length scaling of polar codes, i.e., $n^{1-1/\mu}$. This will conclude that the normalized MIS leakage of the polar secrecy code, built upon capacity-achieving polar codes, scales with $n^{-1/\mu}$, where $\mu$ is the scaling exponent of polar codes for $W$. 

In the next theorem, we characterize bounds on the finite-length scaling behavior of polar secrecy codes, where scaling here refers to the non-asymptotic behavior of both the gap to the secrecy capacity as well as the MIS leakage. More specifically, it is shown that the gap to the secrecy capacity is upper bounded by $O(n^{-1/\muo^*})$ and the normalized MIS leakage is upper bounded by $O(n^{-1/\muo})$. 

\begin{theorem}
\label{thm-scaling}
The polar secrecy code associated with $(\cA,\cR)$, as specified in \eq{AR-def}, has the following properties:
\begin{enumerate}[(i)]
    \item The gap to secrecy capacity:
    \be{cs-gap}
    C_s - R_s \leq \overline{\alpha}^* n^{-1/\muo^*} + o(n^{-1/\muo^*}).
    \ee
    \item Security condition: 
    \be{sec-cond}
    \frac{I(\bU;\bZ)}{k} \leq \overline{\alpha} R_s^{-1} n^{-1/\muo} + o(n^{-1/\muo}).
    \ee
    \item Relibility condition: Bob's probability of error under successive cancellation decoding is upper bounded by $P_e$. 
\end{enumerate}
\end{theorem}
\begin{proof}
Note that $C_s = C(W^*) - C(W)$ for degraded wiretap channels \cite{Wyner}, and by \eq{AR-def}, we have
$$
R_s = \frac{|\cA|}{n} = \frac{|\cG_n(W^*,\epsilon)|}{n} -  \frac{|\cG_n(W,\epsilon)|}{n} 
$$
By invoking this together with the result on the scaling exponent of polar codes, expressed in \eq{polar-scaling2}, for $W^*$ we get:
\begin{align} 
\label{Rs-gap1}
C_s - R_s &= C(W^*) - \frac{|\cG_n(W^*,\epsilon)|}{n} - \bigl(C(W) - \frac{|\cG_n(W,\epsilon)|}{n}\bigr)\\
\label{Rs-gap2}
& \leq \overline{\alpha}^* n^{-1/\muo^*} + o(n^{-1/\muo^*}),
\end{align}
which completes the proof for (i).

To show (ii), we first upper bound the term $\sum_{l \in \mathcal{R}} (1-C_l)$ that appears in the upper bound on the MIS leakage in \Cref{cor-bounds}. Note that for any BMS channel with capacity $C_l$ and Bhattacharyya parameter $Z_l$ we have \cite{Arikan}:
\begin{equation}
\label{thm-scaling-eq1}
C_l + Z_l \geq 1.
\end{equation}
Note also that the finite-length scaling results of Hassani et al. is derived by upper bounding the probability of error as the sum of Bhattacharyya parameters of the good bit-channel. More specifically, the $r$ good bit-channels for transmitting the information bits satisfy the following for the target upper bound on the probability of error $P_e$:
\begin{equation}
\label{thm-scaling-eq2}
\sum_{l \in \mathcal{R}} Z_l \leq P_e.
\end{equation}
Hence, combining \eqref{eq-cor3} together with \eqref{thm-scaling-eq1} and \eqref{thm-scaling-eq2} we arrive at the following upper bound on the MIS leakage:
\begin{equation}
\begin{aligned}
\label{thm-scaling-eq3}
I(\bU; \bZ) & \leq \overline{\alpha} n^{1-1/\muo} + o(n^{1-1/\muo}) + \sum_{l \in \mathcal{R}} Z_l\\
& \leq \overline{\alpha} n^{1-1/\muo} + o(n^{1-1/\muo}) + P_e\\
&= \overline{\alpha} n^{1-1/\muo} + o(n^{1-1/\muo}). 
\end{aligned}
\end{equation}
This together by noting that $R_s = k/n$ completes the proof of (ii). And, finally, (iii) is by the original result on the probability of error of polar codes under the successive cancellation (SC) decoding \cite{Arikan}. In fact, Bob treats the entire $\cA \cup \cR$ as the set of information bits and invokes SC decoder to decode them. The decoded bits on $\cR$ are simply discarded and the decoded bits on $\cA$ are output as the decoded message. 
\end{proof}

\noindent{\textbf{Remark 1.}} The discussions on this section are based upon polar coding with Ar{\i}kan's $2 \times 2$ kernel. There is an extensive research on enhancing the performance of polar codes, including improving their scaling exponent, by utilizing larger $l \times l$ kernels. In particular, it is shown in \cite{fazeli2020binary} that by increasing $l$ one can approach the optimal scaling exponent of $2$. Such results can be incorporated here in a straightforward fashion. More specifically, by invoking the results in \cite{fazeli2020binary}, and by building upon arbitrarily large kernels, one can obtain polar secrecy codes with large kernels whose gap to the secrecy capacity is upper bounded by $O(n^{-\shalf+\epsilon})$ and whose normalized MIS leakage is upper bounded by $O(n^{-\shalf+\epsilon})$, for any fixed $\epsilon >0$. 

\noindent{\textbf{Remark 2.}} Non-asymptotic fundamental limits for wiretap channels are characterized in \cite{Yang}. However, \cite{Yang} considers a non-asymptotic regime where the secrecy leakage, though in terms of the total variation distance (TVD), is treated as a small constant, same as the reliability measure $P_e$. However, the polar secrecy codes in this section operate in a regime where the leakage also approaches zero polynomially in $n$. We are not aware if this regime is addressed in the literature. However, one can naturally expect that there is a sequence of random secrecy codes, together with a converse result stating they are optimal, with the gap to secrecy capacity as well as the normalized MIS leakage of $\Theta(n^{-\shalf})$.

\subsection{Operating above the Secrecy Capacity}

Next, we explore an interesting situation where it is possible to operate slightly above secrecy capacity, with a gap polynomial in $n$ and, of course, vanishing as $n$ grows large while the weak secrecy condition is still satisfied. Note that, in general, one can operate at rates above secrecy capacity at the expense of having a non-zero normalized MIS leakage. In fact, the entire \textit{rate-equivocation region} for degraded wiretap channels is characterized by Wyner \cite{Wyner}. However, to the best of authors' knowledge, how this region exactly scales when operating slightly above (i.e., with a gap vanishing in $n$) the secrecy capacity is not well understood. 

%\noindent{\textbf{Remark 3.}} Note that the polar secrecy code of Section\,\ref{sec:four-A} may already have a rate above the secrecy capacity. This depends on how the scaling exponents of $W$ and $W^*$ compare with each other. Suppose that $\mu > \mu^*$. Then the gap to the secrecy capacity $C_s - R_s$ of the polar secrecy code constructed in Section\,\ref{sec:four-A} becomes negative for large enough $n$'s. This can be observed by invoking \eq{Rs-gap1} together with the definition of scaling exponent in \eq{polar-scaling}, assuming the conjecture on its existence. Note that conditions (ii) and (iii) of \Tref{thm-scaling}, i.e., the security and reliability conditions, respectively, would be still valid. Also, note that in many scenarios we may actually have $\mu > \mu^*$. For instance, given any BMS channel $W^*$, one can degrade it into a binary symmetric channel (BSC) $W$, which is believed to have the worst (largest) scaling exponent. In fact, it is conjectured that binary erasure channel (BEC) has the best scaling exponent while BSC has the worst one \cite{hassani}, for which there is strong empirical evidence. 

Next, we discuss how the polar secrecy code associated with $(\cA,\cR)$, specified in \eq{AR-def}, can be modified to operate above the secrecy capacity while the weak secrecy condition is still satisfied. Let $\delta > \muo^*$ be fixed. For instance any $\delta > 4.714$ would always work \cite{mondelli2016unified}. Then we pick an arbitrary subset $\cR' \subset \cR$ with $|\cR'| = n^{1-\delta}$ and exclude it from $\cR$. More specifically, we set the new $\cR$ and $\cA$ as follows:
\be{AR-def2}
\cR = \cG_n(W,\epsilon) \setminus \cR' ,\ \text{and}\ \cA = \cG_n(W^*,\epsilon) \setminus \cR,
\ee

\begin{theorem}
\label{thm-scaling2}
The polar secrecy code associated with $(\cA,\cR)$, as specified in \eq{AR-def2}, has the following properties:
\begin{enumerate}[(i)]
    \item The gap to secrecy capacity:
    \be{cs-gap2}
    R_s - C_s  = n^{-1/\delta} + o(n^{-1/\delta}).
    \ee
    \item Security condition: 
    \be{sec-cond2}
    \frac{I(\bU;\bZ)}{k} \leq R_s^{-1} n^{-1/\delta} + o(n^{-1/\delta}). 
    \ee
    \item Relibility condition: Bob's probability of error under successive cancellation decoding is upper bounded by $P_e$. 
\end{enumerate}
\end{theorem}

\begin{proof}
To show (i), we have 
\begin{align*} 
R_s - C_s =\ & n^{-1/\delta} + C(W) - \frac{|\cG_n(W,\epsilon)|}{n}\\ &- \bigl(C(W^*)
- \frac{|\cG_n(W^*,\epsilon)|}{n}\bigr)\\
=\ & n^{-1/\delta} + O(n^{-1/\muo^*}) + O(n^{-1/\muo}) \\
=\ & n^{-1/\delta} + o(n^{-1/\delta}). 
\end{align*}
which completes the proof for (i).

To show (ii), by \Cref{cor-bounds} we have 
$$
I(\bU; \bZ) \leq n(C(W)- \frac{|\cR|}{n}) + \sum_{l \in \mathcal{R}} (1-C_l)
$$
Note that the dominating term in $C(W)- \frac{|\cR|}{n}$ is $n^{-1/\delta}$. Also, same is as in the proof of \Tref{thm-scaling}, we have $\sum_{l \in \mathcal{R}} (1-C_l) < P_e$. This completes the proof of (ii). Finally, the argument for part (iii) is exactly same as that of the part (iii) of \Tref{thm-scaling}. 
\end{proof}

%=======================================================================%
%                                                                       %
%     5. Preliminaries                    %              
%                                                                       %
%=======================================================================%
\section{Numerical Results}
\label{sec:five}

Most existing studies on information-theoretic security, such as wiretap channels and secret key generation, primarily focus on asymptotic results. This involves establishing fundamental limits on secrecy rates or devising schemes to approach these limits in the asymptotic sense. However, to validate the practicality of these schemes, it is also critical to present numerical results regarding security guarantees of these schemes. This approach mirrors how performance curves, showing decoder error rates, are frequently showcased and compared to ensure the reliability guarantees of channel coding techniques. There are only a few prior work that have shown finite-length secrecy performance of codes, in terms of MIS leakage or other security metrics. Some numerical results for the polar secrecy codes are presented in \cite{MahdavifarSecrecy} for large block lengths. More recently, \cite{Herfeh} considers the special case where main channel is noiseless and the wiretap channel is BEC and present numerical bounds on the TVD secrecy measure of certain polar and Reed-Muller secrecy codes. In another related work, numerical results on the secrecy performance of randomized convolutional codes for the wiretap channel are presented in \cite{Nooraiepour2017}.

\begin{figure}[t]
\centering
\includegraphics[width=0.99\linewidth]{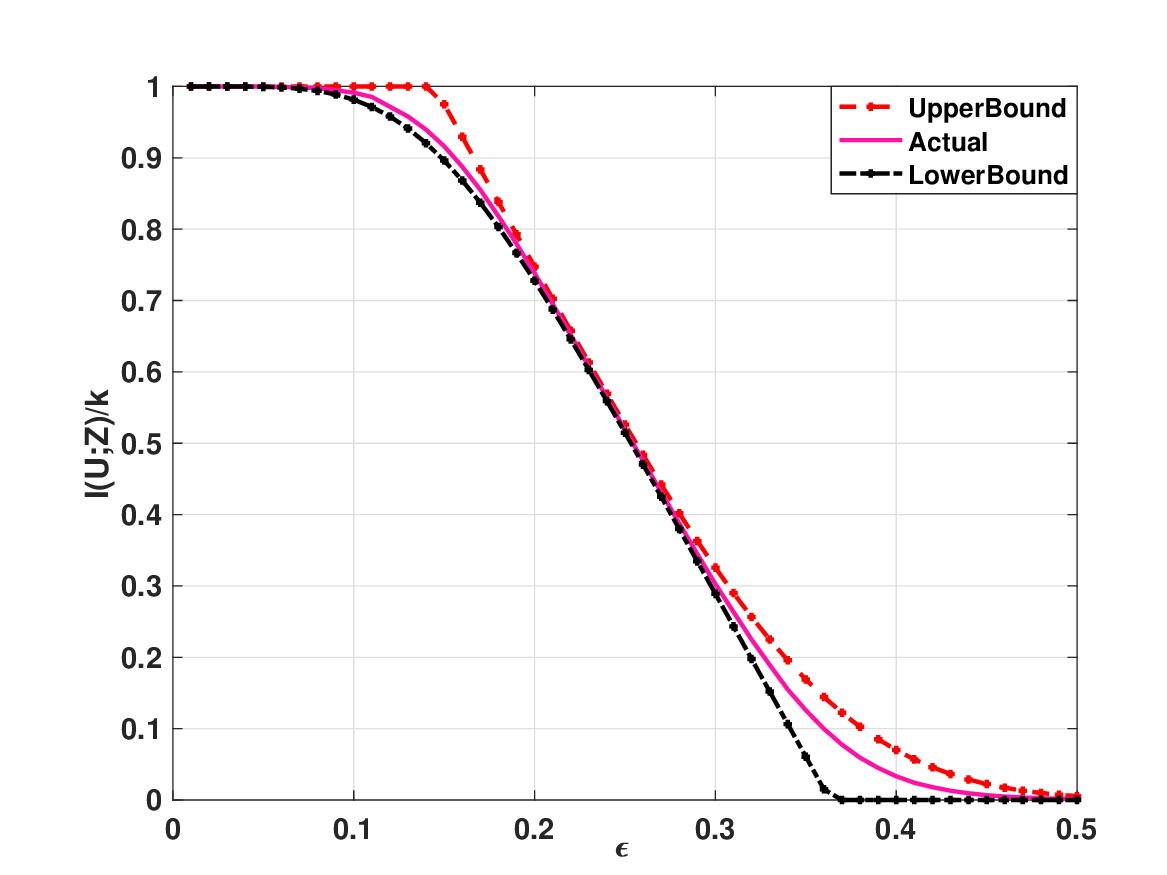}
%{Performance163.png}
\caption{Upper and lower Bounds versus the Monte-Carlo Approximation of the normalized MIS leakage over wiretap BEC$(\epsilon)$.}

\label{fig:256plot}
\end{figure}

The upper and lower bounds presented in this paper, summarized in \Cref{cor-bounds}, not only enables us to characterize the finite-length scaling of polar secrecy codes, as discussed in Section\,\ref{sec:four}, but also to numerically bound the actual MIS leakage in practical settings. Note that directly computing the MIS leakage $I(\bU; \bZ)$ is in general infeasible even for short block lengths $n$. However, the bit-channels' capacities $C_i$'s can be efficiently computed for BECs and can be also well-approximated for general BMS channels \cite{tal2013construct} leading to efficient numerical methods to approximate both the lower and the upper bounds on the MIS leakage.

Furthermore, the actual MIS leakage $I(\bU;\bZ)$ can be approximated using a Monte Carlo method when wiretap channel is a BEC. First, let $\tilde{G}_{(k+r)\times n}$ denote the generator matrix for a polar code associated with $\cA \cup \cR$, where $\cA$ and $\cR$ are the subsets of indices for the information bits and random bits in the considered polar secrecy code. Let $\cA',\cR' \subseteq [k+r]$ denote the indices of information bits and random bits, after discarding the frozen bits. For $i \in \cA'$, let $\bu_i$ denote the indicator vector corresponding to index $i$, i.e., its $i$-th entry is $1$ and the rest of entries are zeros. Let also $\cU_{\cA'}$ denote the subspace of $\F_2^{k+r}$ spanned by $\bu_i$'s, for $i\in \cA'$.  

The Monte Carlo-based approximation of the MIS leakage for wiretap BECs is as follows. One can generate random instances of the erasure patterns for the wiretap channel repeatedly. For each such an instance, let $\tilde{\cG}'$ denote the subspace spanned by the columns of $\tilde{G}$ indexed by non-erasures. Then it is straightforward to observe that $\text{dim}(\tilde{\cG}' \cap \cU_{\cA'})$ is the corresponding \textit{instance} of the MIS leakage. One can average this quantity across a large number of random erasure patterns in order to obtain an approximation for the MIS leakage. 

In Figure\,\ref{fig:256plot} we present numerical results for a polar secrecy code of length $n=256$ with $k=56$, and $r=163$ over wiretap BECs for a range of erasure parameters for the wiretap channel $W$. Both the lower bound of Proposition\,\ref{prop-LB} as well as the upper bound of Proposition\,\ref{prop-UB} on the normalized MIS leakage $I(\bU;\bZ)/k$ are demonstrated together with the approximation of the actual values obtained via the method explained above. The results show that the bounds are actually tight especially for \textit{medium} ranges of the MIS leakage. 

\section*{Acknowledgment}
We would like to thank Hamed Hassani for helpful discussions.

\newpage

\newpage

\bibliographystyle{IEEEtran}
\bibliography{mybib,ref}

% Generated by IEEEtran.bst, version: 1.14 (2015/08/26)
\begin{thebibliography}{10}
\providecommand{\url}[1]{#1}
\csname url@samestyle\endcsname
\providecommand{\newblock}{\relax}
\providecommand{\bibinfo}[2]{#2}
\providecommand{\BIBentrySTDinterwordspacing}{\spaceskip=0pt\relax}
\providecommand{\BIBentryALTinterwordstretchfactor}{4}
\providecommand{\BIBentryALTinterwordspacing}{\spaceskip=\fontdimen2\font plus
\BIBentryALTinterwordstretchfactor\fontdimen3\font minus
  \fontdimen4\font\relax}
\providecommand{\BIBforeignlanguage}[2]{{%
\expandafter\ifx\csname l@#1\endcsname\relax
\typeout{** WARNING: IEEEtran.bst: No hyphenation pattern has been}%
\typeout{** loaded for the language `#1'. Using the pattern for}%
\typeout{** the default language instead.}%
\else
\language=\csname l@#1\endcsname
\fi
#2}}
\providecommand{\BIBdecl}{\relax}
\BIBdecl

\bibitem{Wyner}
A.~D. Wyner, ``The wire-tap channel,'' \emph{The Bell System Technical
  Journal}, vol.~54, no.~8, pp. 1355--1387, 1975.

\bibitem{Maurer}
U.~Maurer and S.~Wolf, ``Information theoretic key agreement: From weak to
  strong secrecy for free,'' \emph{Advances in Cryptology-EUROCRYPT}, p.
  351–368, 2000.

\bibitem{Arikan}
E.~Arikan, ``Channel polarization: A method for constructing capacity-achieving
  codes for symmetric binary-input memoryless channels,'' \emph{IEEE
  Transactions on Information Theory}, vol.~55, no.~7, pp. 3051--3073, 2009.

\bibitem{MahdavifarSecrecy}
H.~Mahdavifar and A.~Vardy, ``Achieving the secrecy capacity of wiretap
  channels using polar codes,'' \emph{IEEE Transactions on Information Theory},
  vol.~57, no.~10, pp. 6428--6443, 2011.

\bibitem{bloch2013strong}
M.~R. Bloch and J.~N. Laneman, ``Strong secrecy from channel resolvability,''
  \emph{IEEE Transactions on Information Theory}, vol.~59, no.~12, pp.
  8077--8098, 2013.

\bibitem{chou2015polar}
R.~A. Chou, M.~R. Bloch, and E.~Abbe, ``Polar coding for secret-key
  generation,'' \emph{IEEE Transactions on Information Theory}, vol.~61,
  no.~11, pp. 6213--6237, 2015.

\bibitem{gunlu2020randomized}
O.~G{\"u}nl{\"u}, P.~Trifonov, M.~Kim, R.~F. Schaefer, and V.~Sidorenko,
  ``Randomized nested polar subcode constructions for privacy, secrecy, and
  storage,'' in \emph{2020 International Symposium on Information Theory and
  Its Applications (ISITA)}.\hskip 1em plus 0.5em minus 0.4em\relax IEEE, 2020,
  pp. 475--479.

\bibitem{wilde2013polar}
M.~M. Wilde and S.~Guha, ``Polar codes for degradable quantum channels,''
  \emph{IEEE Transactions on Information Theory}, vol.~59, no.~7, pp.
  4718--4729, 2013.

\bibitem{jouguet2012high}
P.~Jouguet and S.~Kunz-Jacques, ``High performance error correction for quantum
  key distribution using polar codes,'' \emph{Quantum Information \&
  Computation}, vol.~14, no. 3-4, pp. 329--338, 2014.

\bibitem{fang2022improved}
J.~Fang, Z.~Yi, J.~Li, Z.~Liang, Y.~Wu, W.~Lei, Z.~L. Jiang, and X.~Wang,
  ``Improved polar-code-based efficient post-processing algorithm for quantum
  key distribution,'' \emph{Scientific Reports}, vol.~12, no.~1, pp. 1--11,
  2022.

\bibitem{Sasoglu}
E.~Şaşoğlu and A.~Vardy, ``A new polar coding scheme for strong security on
  wiretap channels,'' in \emph{2013 IEEE International Symposium on Information
  Theory}, 2013, pp. 1117--1121.

\bibitem{hassani}
S.~H. Hassani, K.~Alishahi, and R.~L. Urbanke, ``Finite-length scaling for
  polar codes,'' \emph{IEEE Transactions on Information Theory}, vol.~60,
  no.~10, pp. 5875--5898, 2014.

\bibitem{Korada}
S.~B. Korada, ``Polar codes for channel and source coding,'' EPFL, Tech. Rep.,
  2009.

\bibitem{fazeli2020binary}
A.~Fazeli, H.~Hassani, M.~Mondelli, and A.~Vardy, ``Binary linear codes with
  optimal scaling: Polar codes with large kernels,'' \emph{IEEE Transactions on
  Information Theory}, vol.~67, no.~9, pp. 5693--5710, 2020.

\bibitem{Yang}
W.~Yang, R.~F. Schaefer, and H.~V. Poor, ``Wiretap channels: Nonasymptotic
  fundamental limits,'' \emph{IEEE Transactions on Information Theory},
  vol.~65, no.~7, pp. 4069--4093, 2019.

\bibitem{mondelli2016unified}
M.~Mondelli, S.~H. Hassani, and R.~L. Urbanke, ``Unified scaling of polar
  codes: Error exponent, scaling exponent, moderate deviations, and error
  floors,'' \emph{IEEE Transactions on Information Theory}, vol.~62, no.~12,
  pp. 6698--6712, 2016.

\bibitem{Herfeh}
M.~Shakiba-Herfeh, L.~Luzzi, and A.~Chorti, ``Finite blocklength secrecy
  analysis of polar and {Reed-Muller} codes in {BEC} semi-deterministic wiretap
  channels,'' in \emph{2021 IEEE Information Theory Workshop (ITW)}, 2021, pp.
  1--6.

\bibitem{Nooraiepour2017}
A.~Nooraiepour and T.~M. Duman, ``Randomized convolutional codes for the
  wiretap channel,'' \emph{IEEE Transactions on Communications}, vol.~65,
  no.~8, pp. 3442--3452, 2017.

\bibitem{tal2013construct}
I.~Tal and A.~Vardy, ``How to construct polar codes,'' \emph{IEEE Transactions
  on Information Theory}, vol.~59, no.~10, pp. 6562--6582, 2013.

\end{thebibliography}

\end{document}